\documentclass[a4paper,USenglish,cleveref, autoref, thm-restate, numberwithinsect]{lipics-v2021}

\pdfoutput=1 
\hideLIPIcs  

\title{Investigating Simple Drawings of $K_n$ using \SAT{}}

\author{Helena Bergold}{ }{}{https://orcid.org/0000-0002-9622-8936}{}
\author{Manfred Scheucher}{ }{}{https://orcid.org/0000-0002-1657-9796}{}

\authorrunning{H. Bergold and M. Scheucher} 

\Copyright{Helena Bergold and Manfred Scheucher} 



\category{} 

\relatedversion{}

\supplement{https://github.com/manfredscheucher/rotsys-sat}

\nolinenumbers 

\EventEditors{John Q. Open and Joan R. Access}
\EventNoEds{2}
\EventLongTitle{42nd Conference on Very Important Topics (CVIT 2016)}
\EventShortTitle{CVIT 2016}
\EventAcronym{CVIT}
\EventYear{2016}
\EventDate{December 24--27, 2016}
\EventLocation{Little Whinging, United Kingdom}
\EventLogo{}
\SeriesVolume{42}
\ArticleNo{23}

\usepackage[font=small,labelfont=bf]{caption}
\usepackage[font=small,labelfont=normalfont,labelformat=simple]{subcaption}

\usepackage{multicol}
\usepackage[basic]{complexity}
\usepackage{wrapfig}

\newcommand{\manfred}[1]{\textcolor{blue}{Manfred: #1}}

\usepackage{todonotes}

\newcommand{\calD}{\mathcal{D}}
\newcommand{\calC}{\mathcal{C}}
\newcommand{\calT}{\mathcal{T}}
\newclass{\true}{True}
\newclass{\false}{False}

\newcommand{\simple}{\textnormal{s}}
\newcommand{\convex}{\textnormal{cv}}
\newcommand{\hconvex}{\textnormal{hcv}}
\newcommand{\cmonotone}{\textnormal{cm}}
\newcommand{\strongcmonotone}{\textnormal{scm}}
\newcommand{\gentwisted}{\textnormal{gt}}

\def\obstructionFour{\Pi_4^\textnormal{o}}
\def\obstructionFiveA{\Pi_{5,1}^\textnormal{o}}  
\def\obstructionFiveB{\Pi_{5,2}^\textnormal{o}}  
\def\obstructionconvexFiveA{\Pi_{5,1}^\textnormal{oc}}  
\def\obstructionconvexFiveB{\Pi_{5,2}^\textnormal{oc}}  
\def\obstructionhconvexSix{\Pi_{6}^\textnormal{oh}}

\renewcommand{\paragraph}[1]{\subsubsection*{#1}}

\graphicspath{{figs/}}

\begin{document}

\maketitle

\begin{abstract}
We present a \SAT{} framework which allows to investigate properties of simple drawings of the complete graph~$K_n$ using the power of AI. 
In contrast to classic imperative programming, where a program is operated step by step, our framework models mathematical questions as Boolean formulas which are then solved using modern \SAT{} solvers.
Our framework for simple drawings is based on a characterization via rotation systems and finite forbidden substructures.
We showcase its universality by addressing various open problems, reproving previous computational results and deriving several new computational results.
In particular, we test and progress on several unavoidable configurations such as variants of 
Rafla's conjecture on plane Hamiltonian cycles,
Harborth's conjecture on empty triangles, and
crossing families
for general simple drawings as well as for various subclasses.
Moreover, based our computational results we propose some new challenging conjectures.
\end{abstract}

\section{Introduction}
\label{sec:intro}

In logic and computer science, the Boolean satisfiability problem (\SAT{}) is the problem of determining whether a finite formula on Boolean variables has an assignment such that the formula evaluates to \true. 
\SAT{} is the first problem that was proven to be \NP-complete, which asserts that there is no efficient algorithm for the \SAT{} problem unless \P = \NP. 
Over the last years, however, the area of \SAT{} solving has seen tremendous progress and many problems that seemed to be out of reach a decade ago can now be handled routinely \cite{HKM2016_pythatrip,BalkoValtr2017,Heule2018,BrightCSKG21,HeuleScheucher2024}. 

In this article we present a versatile 
\SAT{} framework for investigating simple drawings of the complete graph~$K_n$. 
In contrast to arbitrary drawings of a fixed graph, which can be arbitrarily complex,  \emph{simple} drawings can be encoded by a finite Boolean formula.  
A drawing of a graph is \emph{simple} if every pair of edges has at most one common point, which is either a crossing or a common endpoint.

Many properties 
that are in focus of active research 
such as plane Hamiltonian cycles or empty triangles
do not depend on the actual drawing but only on the underlying combinatorics of the drawing.
More specifically, the purely combinatorial information about the pairs of crossing edges and the cyclic orders around the vertices -- the so-called \emph{rotation system} of the drawing -- is sufficient to study many problems. 

A crucial ingredient of our framework is a theorem by Kyn\v{c}l~\cite{Kyncl2020}
together with the computational result by \'Abrego et al.\ \cite{AbregoAFHOORSV2015}. 
It gives a compact characterization which systems of cyclic permutations are in fact a rotation system of a simple drawing.
We give the precise definitions and the characterization in \cref{sec:prelim}. 
In \Cref{sec:all_encoding} we describe the encoding of the combinatorial structures such that the solutions are in correspondence with rotation systems. Since the default input format of \SAT{} solvers are Boolean formulas in conjunctive normal form (CNF), 
we present a list of clauses which are then conjugated (logical “AND”) to one formula.
A clause is a disjunction (logical “OR”) of literals, 
where a literal is a Boolean variable or its negation. 
For more background see the standard textbook~\cite{HandbookSatisfiablity2009}.
In \Cref{app:characterizationRS} we moreover develop a SAT encoding for checking whether a system comes from a drawing and give an alternative proof of the characterization of \'Abrego et al.\ \cite{AbregoAFHOORSV2015}.

Additionally to the general setting of simple drawings, 
our framework allows 
to restrict the search space
to certain subclasses that appear in active research.
Among them are 
convex, hereditary convex, c-monotone, strongly c-monotone or generalized twisted drawings.
All of these subclasses 
can be characterized 
in terms of the rotation system and yield a reasonably compact \SAT{} encoding; see \cref{sec:subclasses} for more details.

In \cref{sec:applications}, we showcase the power of our framework. 
We formulate various open problems related to simple drawings of the complete graph as a CNF and test several conjectures for small values. 
In particular, we focus on the existence of plane Hamiltonian cycles and other plane substructures related to Rafla's conjecture  (\cref{sec:applications_planar,ssec:kyccles}), 
the existence of uncrossed edges and crossing families (\cref{sec:uncrossededges,sec:crossingfamilies}), 
unavoidable subdrawings of the complete graph (\cref{sec:unavoid}), 
and 
the number of empty triangles (\cref{sec:emptytriang}). 

The previously best computational results were obtained using enumerative and brute-force approaches~\cite{AbregoAFHOORSV2015}.
However, as there exist approximately 7 billion rotation systems of~$K_9$ and the number of rotation systems of~$K_n$ grows as $2^{\Theta({n^4})}$ \cite{Kyncl2009}, 
an enumerative approach for $n=10$ will not be possible in reasonable time with contemporary computers.
Our novel \SAT{}-based framework certainly surpasses these previous approaches as it allows to make investigations for 10 to 15 vertices with reasonably small resources -- of course depending on the desired property. 
Moreover, our framework  gives an excellent computational basis for searching for drawings on 20 or more vertices with certain properties (cf.\ \cref{sec:unavoid}),
even though a complete search seems unlikely at the current stage.

\section{Preliminaries}
\label{sec:prelim}

Various properties of a simple drawing
only depend on the combinatorics of the drawing. 
For example, to determine the existence of plane substructures, it is sufficient to know which pairs of edges cross -- the actual routing of edges does not play a role.
More generally, we will make extensive use of the cyclic orders of edges around the vertices, which in particular capture the information about crossing edge pairs. 
In this article we only consider properties which do not depend on the choice of the unbounded cell.
Hence, in the following we often do not distinguish between drawings on the sphere and drawings in the plane. 
	
For a given simple drawing $\calD$ and a vertex $v$ of $\calD$, the cyclic
order~$\pi_v$ of incident edges in counterclockwise order around~$v$ is called
the \emph{rotation of $v$} in $\calD$. The collection of rotations of all
vertices is called the \emph{rotation system} of~$\calD$. In the case of 
the complete graph~$K_n = (V,E)$, the rotation of a vertex $v$ is a cyclic
permutation on the remaining $n-1$ vertices $V \backslash \{v\}$.

A \emph{pre-rotation system} on $V$ consists of cyclic permutations $\pi_v$ on
the elements $V \backslash \{v\}$ for all $v \in V$. 
A pre-rotation system $\Pi = (\pi_v)_{v \in V}$ is \emph{drawable} if there is a
simple drawing of the complete graph with vertices $V$ such that its rotation
system coincides with~$\Pi$. Two pre-rotation systems are \emph{isomorphic}
if they are the same up to relabelling and reflection (i.e., all cyclic orders
are reversed). 
Two simple drawings are \emph{weakly isomorphic}
if their rotation systems are isomorphic. 
Note that there exist several types of isomorphisms in literature.

\begin{figure}[htb]
    \centering
    \includegraphics[scale = 0.9]{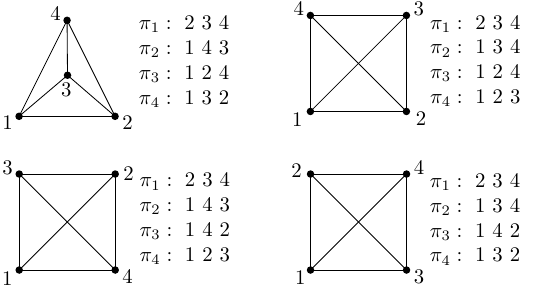}
    \caption{
    The four rotation systems on $4$ elements with their drawings. The second, third, and fourth are isomorphic.
    The first one represents the second isomorphism class. }
    \label{fig:rs_n4_noniso}
\end{figure}

On four vertices there are three non-isomorphic pre-rotation systems. 
The $K_4$ has exactly
two non-isomorphic simple drawings on the sphere: the drawing with no crossing
and the drawing with one crossing;
see \cref{fig:rs_n4_noniso}.
Hence, the two corresponding pre-rotation systems are drawable, and the third pre-rotation system is an obstruction to drawability. 
It is denoted by $\obstructionFour$ and described 
in \cref{fig:rotsys_obstructions}.

\begin{figure}[htb]
    \centering
    \includegraphics[scale = 0.9]{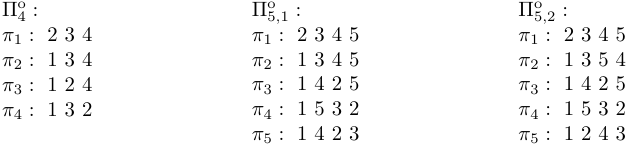}
    \caption{The three obstructions $\obstructionFour$, $\obstructionFiveA$, and $\obstructionFiveB$ for rotation systems.}
    \label{fig:rotsys_obstructions}
\end{figure}
	
For a pre-rotation system $\Pi = (\pi_v)_{v \in V}$ and a subset of the elements 
$I \subseteq V$, the \emph{sub-configuration induced by $I$} is
$\Pi|_I = (\pi_v|_I)_{v \in I}$, where $\pi_v|_I$ denotes the cyclic permutation
obtained by restricting $\pi_v$ to $I \backslash \{v\}$. A pre-rotation
system $\Pi$ on $V$ \emph{contains} $\Pi'$ if there is an induced
subconfiguration $\Pi|_I$ with $I \subseteq V$ isomorphic to~$\Pi'$.
A pre-rotation system not containing $\Pi'$ is called
\emph{$\Pi'$-free}.

A crossing pair of edges involves four vertices. 
By studying drawings
of $K_4$ we learn that a crossing pair of edges
can be identified from the underlying rotation system. 
Hence, the pairs of crossing edges in a drawing of $K_n$ are fully determined by the rotation system.

\begin{lemma}			
    \label{observation:basics}
    The following two statements hold:
    \begin{enumerate}[(i)]
        \item \label{item:obstructionfour_notdrawable}
        A pre-rotation system containing $\obstructionFour$ is not drawable.
    			
        \item 
        \label{item:obstructionfour_crossingsdetermined}
        Let $\Pi$ be a $\obstructionFour$-free pre-rotation system on $[n]$. The
        subconfiguration induced by a 4-element subset is drawable and determines which pairs of
        edges cross in the drawing.
    \end{enumerate}
\end{lemma}

Part \eqref{item:obstructionfour_crossingsdetermined} of the lemma allows to talk about the crossing pairs of edges of a $\obstructionFour$-free pre-rotation system, even if there is no associated drawing.
	\'Abrego et al.~\cite{AbregoAFHOORSV2015} generated all pre-rotation systems
	for up to $9$ vertices and used a drawing program based on back-tracking to
	classify the drawable ones. 
	
	\begin{restatable}[\cite{AbregoAFHOORSV2015}]{proposition}{propclassificationRS}
		\label{proposition:rotsys_classification_n6}
		A pre-rotation system on $n \le 6$ elements
		is drawable if and only if it does not contain
		$\obstructionFour$, $\obstructionFiveA$, or $\obstructionFiveB$
		(cf.\ \cref{fig:rotsys_obstructions}) as a subconfiguration.
	\end{restatable}

	Moreover,
	Kyn\v{c}l showed that a pre-rotation system is drawable if and only if all induced 4-, 5-, and 6-element subconfigurations are drawable \cite[Theorem 1.1]{Kyncl2020}.
	Together with \cref{proposition:rotsys_classification_n6} this yields the following characterization:
	
\begin{restatable}{theorem}{thmclassificationRS}
    \label{theorem:rotsys_5tuples_characterization}
    A pre-rotation system on $n$ elements is drawable if and only if
    it does not contain $\obstructionFour$, $\obstructionFiveA$ or $\obstructionFiveB$ (cf.\ \cref{fig:rotsys_obstructions})
    as a subconfiguration.
\end{restatable}

While the focus of this article is on the structural investigation of simple drawing of~$K_n$,
we provide an independent proof of \cref{proposition:rotsys_classification_n6}
in \cref{app:characterizationRS}, 
which
utilizes \SAT{} solvers on two levels. First we use a \SAT{} solver to enumerate all non-isomorphic pre-rotation systems. Then, for each pre-rotation system, we create a \SAT{} instance to test its drawability.
More specifically, 
to decide whether a given pre-rotation system
yields a drawing, 
we read the crossings along each edge (cf.\ \cref{observation:basics})
and then use \SAT{} to find an ordering of the crossings along all edges so that the obtained subdivided graph corresponds to a planarization of a drawing.
To ensure planarity, we utilize Schnyder's
characterization of planar graphs \cite{Schnyder1989}; see \cref{app:planar_enc}.
For further details and more planarity encodings suitable for the \SAT{}-based investigations of planar graphs we refer the interested reader to~\cite{KirchwegerSS2023}.

\section{SAT Encoding}
\label{sec:all_encoding}

We develop a versatile SAT framework for rotation systems of simple drawings of~$K_n$.
It comes with many optional parameters that make it possible  
to search for examples with (a combination of) certain properties
and to restrict the search to certain subclasses.
To verify conjectures for small $n$ one can then create a specific CNF instance to search for a counterexample and use a SAT solver such as CaDiCaL \cite{Biere2019} to prove unsatisfiablity.
It is worth noting that unsatisfiablity can also be certified with an independent proof-verification tool such as drat-trim \cite{WetzlerHeuleHunt2014}.
In the following we describe the basic encodings.

\subsection{Pre-Rotation Systems}
\label{ssec:preRS}
A pre-rotation system on $[n]$ consists of $n$ cyclic permutations $\pi_a$ on $[n] \setminus \{a\}$.
One way to encode the cyclic permutations in terms of a CNF is to use ordinary permutations.
For this we introduce a Boolean variable $X_{aib}$ for every pair of distinct vertices $a, b \in [n]$ and every index $i \in [n-1]$, to indicate whether $\pi_a(i)=b$.
To ensure that for every $a \in [n]$ the variables $X_{aib}$ indeed model a  permutation $\pi_a$, we 
introduce clauses such that for every $i \in [n-1]$ there is exactly one $b \in [n] \setminus \{a\}$ for which $X_{aib}$ is \true.
More formally, we use the clauses
\begin{itemize}
    \item $\bigvee_{b:b \neq a} X_{aib}$ \quad for every $a \in [n]$ and $i \in [n-1]$, \quad and 
    \item $\neg X_{aib_1} \vee \neg X_{aib_2}$ \quad for every distinct $a,b_1,b_2 \in [n]$ and $i \in [n]$.
\end{itemize}
The first clause ensures that for fixed $a$ and $i$ at least one of the variables $X_{aib}$ is \true, and the second implies that at most one is \true. 
Since we deal with cyclic permutations, we assume without loss of generality that the first element in a permutations is the smallest. Therefore, we add the unit clauses 
\begin{itemize}
    \item $X_{112}=\true$ \quad and \quad 
     $X_{k11}=\true$ \quad for $k \geq 2$.
\end{itemize}
Since the $X$-variables only give us a ``global'' description of the entire rotation system, we introduce additional auxiliary variables to describe the rotation system around each vertex in more detail.
We introduce an auxiliary variable $Y_{abcd}$ for every four distinct vertices $a,b,c,d \in [n]$
to indicate whether $b,c,d$ appear in counterclockwise order in the rotation of~$a$. 
Note that
precisely one of two variables $Y_{abcd}$ and $Y_{abdc}$ is \true, and that $Y_{abcd}=Y_{acdb}=Y_{adbc}$. 
To synchronize the $X$-variables with the $Y$-variables,
we assert for every distinct $a,b,c,d \in [n]$ and $i,j,k \in [n-1]$ that
\begin{itemize}
    \item $(X_{aib} \wedge X_{ajc} \wedge X_{akd}) \rightarrow \phantom{\neg} Y_{abcd}$ 
    \quad if $i<j<k$ or $k<i<j$ or $j<k<i$,
    \quad and 
    \item $(X_{aib} \wedge X_{ajc} \wedge X_{akd}) \rightarrow \neg Y_{abcd} $
    \quad if $i<k<j$ or $k<j<i$ or $j<i<k$.
\end{itemize}
To model this as a CNF, we use that $(A_1 \wedge \ldots \wedge A_k) \rightarrow B$ is logically equivalent to $\neg A_1 \vee \ldots \vee \neg A_k \vee B$.
While the above encoding derives the values of the $Y$-variables from the $X$-variables, we describe in the following how to directly encode the cyclic permutations only in terms of the $Y$-variables. 
In fact, when we used the \SAT{} solver CaDiCaL we observed that the following clauses were learned after some time.
By adding these clauses to the CNF,
the solver performed significantly faster. 

A mapping $\sigma: {X \choose 3} \to \{+,-\}$ encodes a \emph{cyclic permutation} if for all distinct four elements $i,j,k,\ell$ with $i<j<k<\ell$ the sign sequence 
$(\sigma(i,j,k),\sigma(i,j,\ell),\sigma(i,k,\ell),\sigma(j,k,\ell))$ is one of the following six:
    $\{(+,+,+,+), (+,+,-,-), (+,-,-,+), (-,-,-,-), (-,-,+,+), (-,+,+,-) \}$
In other words the following ten sign patterns are forbidden: 
\begin{align*}
    \{&(+,-,+,-),(-,+,-,+),(+,-,-,-),(-,+,-,-),(-,-,+,-),\\
    &(-,-,-,+),(-,+,+,+),(+,-,+,+),(+,+,-,+),(+,+,+,-) \}
\end{align*} 
To forbid the pattern $(+,-,+,-)$ for each $a \in [n]$ and all $b<c<d<e $ from $[n] \setminus \{a\}$, we use the following clause: 
\begin{itemize}
    \item $\neg Y_{abcd} \vee  Y_{abce} \vee \neg Y_{abde} \vee Y_{acde} $ 
\end{itemize}
The other patterns are forbidden analogously.
To further improve the encoding, these ten
clauses of length~4 can be replaced by the following eight clauses of length~3:
\begin{multicols}{2}
\begin{itemize}
    \item $ \neg Y_{abcd} \vee \phantom{\neg} Y_{abce} \vee \neg Y_{abde}$
    \item $ \phantom{\neg} Y_{abcd} \vee  \neg Y_{abce} \vee \phantom{\neg} Y_{abde}$
    \item $ \neg Y_{abce} \vee \phantom{\neg} Y_{abde} \vee \neg Y_{acde}$
    \item $ \phantom{\neg} Y_{abce} \vee  \neg Y_{abde} \vee \phantom{\neg} Y_{acde}$
    \item $ \neg Y_{abcd} \vee \neg Y_{abde} \vee \phantom{\neg} Y_{acde}$
    \item  $ \phantom{\neg} Y_{abcd} \vee \phantom{\neg} Y_{abde} \vee \neg Y_{acde}$
    \item $ \neg Y_{abcd} \vee \phantom{\neg} Y_{abce} \vee \phantom{\neg} Y_{cbde}$
    \item $ \phantom{\neg} Y_{abcd} \vee \neg Y_{abce} \vee \neg Y_{cbde}$
\end{itemize}
\end{multicols}

\subsection{Rotation System}
\label{ssec:RS}
To restrict the search space to
drawable pre-rotation systems,
that is,
rotation systems of the complete graph~$K_n$,  
we introduce clauses to forbid the drawability-obstructions. 	\cref{theorem:rotsys_5tuples_characterization}
asserts that a pre-rotation system is drawable if and only if none of the obstructions $\obstructionFour$, $\obstructionFiveA$ and $\obstructionconvexFiveB$ (depicted in \cref{fig:rotsys_obstructions}) occur as a substructure.
To ensure that the obstructions~$\obstructionFour$ or its reversed does not occur in any induced subconfiguration,
the two clauses
\begin{itemize}
    \item $\neg Y_{abcd} \vee
    \neg Y_{bacd} \vee
    \neg Y_{cabd} \vee
    \neg Y_{dacb}$ \quad and
    \quad  $\phantom{\neg} Y_{abcd} \vee
    \phantom{\neg} Y_{bacd} \vee
    \phantom{\neg} Y_{cabd} \vee
    \phantom{\neg} Y_{dacb}$
\end{itemize}
must be fulfilled for every four distinct vertices $a,b,c,d \in [n]$.
In an analogous manner, 
we can forbid $\obstructionFiveA$ and $\obstructionFiveB$ as subconfigurations.
In total, we have $\Theta(n^5)$  clauses to assert that solutions of the CNF are in one-to-one correspondence with rotation systems on $[n]$ elements.

By enabling the flags \verb|-v4| and \verb|-v5| one can disable the above clauses for valid 4-tuples and valid 5-tuples, respectively, and enumerate pre-rotation systems.

\subsection{Symmetry Breaking}

Since most properties discussed in this article are invariant to relabeling and reflection,
we attempt to break these symmetries to get a one-to-one correspondence between the solutions of the CNF and 
the isomorphism classes
of (pre-)rotation systems.

Since there are $n!$ relabelings of the vertices and the possibility of reflection,
there are up to $2n!$ (pre-)rotation systems in a weak isomorphism class. 
We define a unique representative of every isomorphism class. 
For this, observe that
the $X$-variables of
a pre-rotation system~$\Pi$ can be read
as an $n \times (n-1)$ matrix, 
where the $a$-th row encodes the permutation of the $a$-th element.
By concatenating the rows of the matrix, we consider a vector of length $n \cdot (n-1)$.
The \emph{lexicographic minimal} vector among all relabelings and reflections is our representative. 

At first glance, it seems that 
we have to check all $n!$ relabellings and their reflections
to determine whether a pre-rotation system~$\Pi$
is a lexicographic minimum.
However, the lexicographic minimum has to be \emph{natural},
that is, the rotation around the first vertex is the identity permutation. To assert that a pre-rotation system is natural,
we use the unit clauses:
\begin{itemize}
    \item $Y_{1abc} = \true$ for all $a,b,c \in [n] \setminus \{1\}$ such that $a<b<c$.
\end{itemize} 
Being natural is only a necessary condition but not sufficient; see for example 
the three natural rotation systems of $K_4$ 
with a crossing depicted in \cref{fig:rs_n4_noniso}.
However, since the choice of the first and second vertex in a natural labeling fully determines the first row and thus the full permutation of the vertices, we only have to test $n(n-1)$ relabellings and their reflections.
Recall that we do not actually run an algorithm to perform those tests 
but ingeniously add constraints to the CNF so that all solutions fulfill the desired properties.

Besides the identity permutation on $[n] \backslash \{1\}$ in the first row of the matrix in a natural pre-rotation system, w e further assume that the first element in every tow is the smallest one of the cyclic permutation. 
To ensure that it is a lexicographical minimum,
we introduce auxiliary variables and clauses to compute the 0-1-vector of all $2n(n-1)-1$ reflections and relabelings.
Additional constraints in terms of the $X$-variables assert that the original vector is indeed the lexicographically smallest. 
Since the encoding in terms of Boolean variables is quite technical, we omit the details and refer the interested reader to the source code~\cite{supplemental_data_unblind}.

While our framework by default restricts the search space to natural rotation systems (this can be disabled via the flag \verb|-nat|), 
the restriction to lexicographic minima needs to be explicitly enabled using the flag \verb|-lex| as it seems to be quite expensive in practice. Even though it should reduce the search space by a factor of up to $2n(n-1)$ in theory,  \SAT{} solvers are slowed down in practice due to the large amount of auxiliary variables and clauses.

\subsection{Crossings}
\label{sec:crossingsencoding}
Many open questions related to simple drawings concern plane substructures and crossings.
To encode crossings in terms of a CNF 
we 
introduce auxiliary variables $C_{abcd} = C_{ef}$ to indicate whether two non-adjacent edges $e=\{a,b\}$ and $f=\{c,d\}$ cross 
(cf.\ \cref{observation:basics}\eqref{item:obstructionfour_crossingsdetermined}).
If the drawing of $K_4$ with vertices $\{a,b,c,d\}$ 
is crossing-free, the rotation system is unique up to relabelling or reflection.
Otherwise, if there is a crossing, we have one of the following: $ab$ crosses $cd$, $ac$ crosses $bd$, or $ad$ crosses $bc$ and this is fully determined by the rotation system.
For each of these three cases, there are two subcases.
If $ab$ crosses $cd$, then the directed edge $\overrightarrow{cd}$ either traverses $\overrightarrow{ab}$ from the left to the right or vice versa.
Each subcase corresponds to a unique rotation system.
The three cases are depicted in \cref{fig:rs_n4_noniso}. The other cases are the reversed rotation systems which correspond to the reflected drawings.
To indicate the subcase, we
introduce an auxiliary variables $D_{abcd}$ and the clauses $C_{abcd} = D_{abcd} \vee D_{abdc}$.

To prove that every simple drawing of $K_n$ has a particular structure, we encode the contraposition. 
Solutions of the CNF encode potential examples which violate this structure. If the CNF is unsatisfiable, the statement follows. 
For more details see~\cref{app:encodings_planesubstructures}.

\section{Subclasses}
\label{sec:subclasses}

Our framework comes with plenty of optional parameters to restrict the search space and/or to require certain properties. In the following we discuss some subclasses that can be fully characterized in terms of the rotation system and moreover yield a reasonably sized encoding in terms of a~CNF.
For a detailed overview of subclasses studied in literature and their relationships 
we refer to \cite[Figure~7]{aov-2024-tcfhcsdcg}.

\subsection{Convex and H-convex Drawings}

Convex drawings are a subclass between \emph{geometric} drawings where all edges are drawn as straight-line segments and simple drawings. As indicated by the name, they generalize the classical notion of convexity.
To give a formal definition, we need to look at the \emph{triangles}, i.e., 
the subdrawings induced by triples of vertices. 
Since in a simple drawing of~$K_n$ the edges of a triangle do not
cross, a triangle partitions the plane (resp.\ sphere) into exactly two
connected components. 
The closures of these components are the two
\emph{sides} of the triangle. A side $S$ of a triangle is \emph{convex} if
every edge that has its two vertices in~$S$ is fully drawn inside~$S$. A
simple drawing of $K_n$ is \emph{convex} if every triangle has a convex
side. Moreover, a drawing is \emph{h-convex} (short for hereditary
convex) 
if we can choose a convex side $S_T$ for every triangle $T$ such that,
for every triangle $T'$ contained in $S_T$, it holds $S_{T'} \subseteq S_T$.
For further aspects of the
convexity hierarchy we refer to
\cite{ArroyoMRS2017_convex,ArroyoMRS2017_pseudolines,BFSSS_TDCTCG_2022}.

Arroyo et al.~\cite{ArroyoMRS2017_convex} showed that convex and h-convex
drawings can be characterized via finitely many forbidden subconfigurations.
A simple drawing is \emph{convex} if and only if it does not contain
$\obstructionconvexFiveA$ or $\obstructionconvexFiveB$ (cf.
\cref{fig:rotsys_obstructions_convex}) as a subconfiguration.
Moreover, a convex drawing is \emph{h-convex} if and only if it does not
contain $\obstructionhconvexSix$ (cf.
\cref{fig:rotsys_obstructions_hconvex}) as a subconfiguration.

\begin{figure}[htb]
    \centering
    \includegraphics[page=2, scale = 0.9]
    {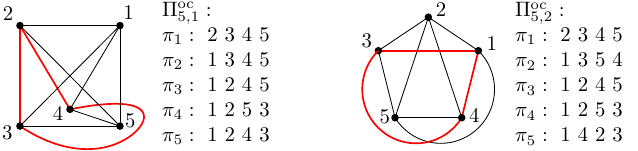}
    \caption{The two obstructions $\obstructionconvexFiveA$ (left) and $\obstructionconvexFiveB$ (right) for convex drawings. A triangle which is not convex is highlighted in red.}
    \label{fig:rotsys_obstructions_convex}
\end{figure}

\begin{figure}[htb]
    \centering
    \includegraphics[page=3,scale=0.9]
    {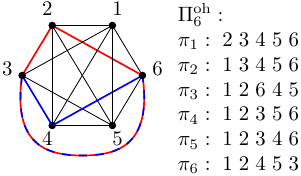}
    \caption{The obstruction $\obstructionhconvexSix$ for h-convex drawings.
    The convex side of the red triangle is the bounded side and for the blue triangle its the unbounded side.  
    }
    \label{fig:rotsys_obstructions_hconvex}
\end{figure}

Our framework by default enumerates rotation systems of simple drawings of $K_n$.
To restrict the search space to convex drawings and h-convex drawings, one can use the flag \verb|-c| and \verb|-hc|, respectively.
With these flags, $\Theta(n^5)$ (resp.\ $\Theta(n^6)$) clauses are added to forbid the configurations $\obstructionconvexFiveA$ and $\obstructionconvexFiveB$ (resp.\ $\obstructionhconvexSix$) as induced subconfigurations,
similar as in~\Cref{ssec:RS}.

\subsection{C-Monotone and Strongly C-Monotone Drawings}

A simple drawing $\calD$ in the plane is \emph{c-monotone} if there is a point $O$ such that every ray emanating from $O$ intersects every edge of $\calD$ at most once. 
Moreover, $\calD$ is \emph{strongly c-monotone} if, for every vertex~$v$, there is a ray which does not intersect edges incident to~$v$. 
More generally, we say that a rotation system is (strongly) c-monotone if it has a (strongly) c-monotone drawing.
Even though on the first glance this definition seems to strongly depend on the actual drawing, 
Aichholzer et al.~\cite{AichholzerGTVW22twisted} proved the following combinatorial characterization: 
a rotation systems on elements $V=\{a_1,\ldots,a_n\}$ is c-monotone if and only if it can be extended by two additional elements $b_1,b_2$ such that 
no two edges $b_1 a_i$ and $b_2 a_j$ cross.
Moreover, it is strongly c-monotone if for every $i$ there exists $j=j(i)$ such that the star centered at $a_i$ does not cross the edges $a_jb_1$ and $a_jb_2$.
In our framework, the flags \verb|-cm| and \verb|-scm| restrict to c-monotone and strongly c-monotone rotations systems, respectively.

\subsection{Generalized Twisted Drawings}

A c-monotone drawing of $K_n$ is \emph{generalized twisted} if there is one ray emanating of the point~$O$ that intersects all edges of the drawing.
This can be implemented in a similar way as above. 
As shown by Aichholzer et al.~\cite{agtvw-2023-crsgtd5t}, generalized twisted drawings with $n \geq 7$ are exactly those simple drawings, where every subconfiguration on 5 elements is isomorphic to~$\obstructionconvexFiveA$. 
This combinatorial description can be encoded in terms of rotation systems by forbidding all other 4 possibilities for subconfigurations on 5 elements. 
This characterization can be activated in our framework using the flag \verb|-gt|.

\section{Applications and Showcases of our Framework}
\label{sec:applications}

\subsection{Plane Hamiltonian Substructures}
\label{sec:applications_planar}

One prominent conjecture on plane substructures in simple drawings of $K_n$ 
is by Rafla.

\begin{conjecture}[{Rafla \cite{Rafla1988}}]
    \label{conjecture:rafla}
    Every simple drawing of $K_n$ with $n \ge 3$ contains a plane Hamiltonian cycle.
\end{conjecture}

Rafla verified the conjecture for $n \le 7$. Later \'Abrego et al.\ \cite{AbregoAFHOORSV2015} enumerated all rotation systems for \mbox{$n \le 9$} and verified Rafla's conjecture. 
The conjecture was recently proven for the subclasses of cylindrical and strongly c-monotone drawings~\cite{aov-2024-tcfhcsdcg}, for convex drawings~\cite{BFROS2024}, and for generalized twisted drawings with an odd number of vertices~\cite{AichholzerGTVW22twisted}.
With our \SAT{} framework, 
we got additional computational results.

\begin{theorem}
\cref{conjecture:rafla} is true for all simple drawings with $n \le 10$ and generalized twisted drawings with $n \leq 13$.
\end{theorem}

We searched for a counterexample to the conjecture with the parameter \verb|-HC|.
For $n=8,9,10$  the solver CaDiCaL
 took about 25 CPU seconds, 80 CPU minutes, and 3 CPU days, respectively, to show unsatisfiability.
In the setting of generalized twisted drawings, the computations for $n=12$ took about 3 CPU hours and the case $n=13$ is covered in~\cite{AichholzerGTVW22twisted}. 
Furthermore, we exported a CNF file with parameter \verb|-o| and ran the stand-alone executable of CaDiCaL to create a DRAT certificate. 
The certificate for $n=10$ is about 78GB and 
the verification of the unsatisfiability took about 6 CPU days with the proof checking tool DRAT-trim  \cite{WetzlerHeuleHunt2014}.
The resources used for $n \le 9$ are negligible.

\medskip

\goodbreak 

Recently, Aichholzer et al.\ \cite{aov-2024-tcfhcsdcg} provided a strengthening of Rafla's conjecture:
\begin{conjecture}[{
\cite[Conjecture~1.2]{aov-2024-tcfhcsdcg}}]
\label{conj:all_pairs_HP}
    For every simple drawing of $K_n$ and for every choice of two vertices $a,b$, there is a plane Hamiltonian path from $a$ to $b$. 
\end{conjecture}

Similar to Rafla's conjecture, this conjecture was verified using the database of rotation systems for $n \leq 9$. Moreover, it was proven for strongly c-monotone drawings, cylindrical, and convex drawings, see \cite{aov-2024-tcfhcsdcg,BFROS2024}. 
To test 
\cref{conj:all_pairs_HP} 
with our framework, we use the flag \verb|--forbidAllPairsHP|. The computations for $n = 10$ took about 4 CPU days, and for the restricted setting of generalized twisted drawings about 5 CPU hours for $n =13$. 

\begin{theorem}
\cref{conj:all_pairs_HP} holds for $n \le 10$ for simple drawings of $K_n$ and for $n \leq 13$ for generalized twisted drawings.
\end{theorem}

Another strengthening of Rafla's conjecture was recently proposed by Bergold et al.~\cite{BFROS2024}:
\begin{conjecture}[{
\cite[Conjecture~5.2]{BFROS2024}}]
\label{conj:2n_minus_3}
    Every simple drawing of $K_n$ contains a plane subdrawing with $2n-3$ edges which contains a Hamiltonian cycle. 
\end{conjecture}

As the main result of their article, they prove the conjecture for the subclass of convex drawings. 
We here use our \SAT{} framework with parameter \verb|-HC+| to verify the conjecture for simple drawings up to $n \le 9$ vertices. The computations took about 3 CPU days.

\begin{theorem}
\cref{conj:2n_minus_3} holds for $n \le 9$.
\end{theorem}

While for the previous two conjectures, we continued the search for subclasses and larger~$n$, the number of clauses in this encoding (cf.\ \cref{app:encodings_planesubstructures}) 
grew too fast to continue. 

\medskip

As shown in~\cite{BFROS2024}, there are several variant of prescribing structures for plane Hamiltonian subgraphs. However for simple drawings, we cannot even assign an edge to be contained in a Hamiltonian cycle. 
A weakening of this question is prescribing edges such that there is a Hamiltonian cycle which does not cross these edges (but not necessarily contain them).
A natural question is to
prescribe plane matchings  
and
ask for a plane Hamiltonian cycle which together with the
matching builds a plane Hamiltonian subgraph, 
i.e., the edges of the matching are
not crossed by the Hamiltonian cycle (but possibly contained).
For the setting of a geometric drawings of $K_n$
Hoffmann and T\'oth
\cite{HoffmannToth2003} showed that
for every plane perfect matching~$M$ there
exists a plane Hamiltonian cycle that 
does not cross any edge from~$M$. See also \cite{Hoffmannphd05} for a version with not necessarily perfect matchings. 

While the statement does not generalize to the simple drawings and not even to strongly c-monotone
(see \cref{fig:HTP_nonconvex_counter}),
our computations suggest that
it is true for convex drawings. 
Hence, we dare another strengthening of Rafla's conjecture:
	
\begin{conjecture}
\label{conjecture:hoffmanntoth_convex}
    For every plane matching~$M$
    in a convex drawing of~$K_n$ 
    there exists a plane Hamiltonian cycle 
    that does not cross any edge from~$M$.
\end{conjecture}

\begin{figure}[htb]
    \centering
    \includegraphics{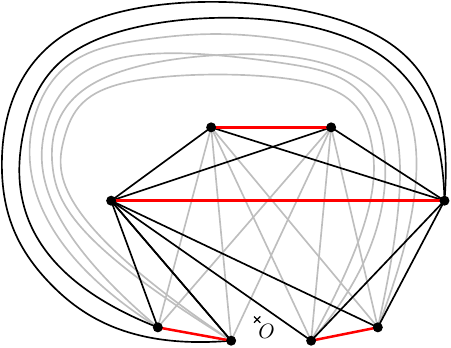}
    \caption{A perfect matching (red)
    in a non-convex and strongly c-monotone drawing of~$K_8$, 
    which does not contain a plane Hamiltonian cycle not crossing the matching edges. Edges which cross matching edges are marked in grey. 
    }
    \label{fig:HTP_nonconvex_counter}
\end{figure}

Using our framework and the parameter \verb|-HT+ k|, where $k \in \{1,\ldots,\lfloor n/2 \rfloor \}$ denotes the size of the matching, it took about 12 CPU days in total to verify the conjecture for
$n \leq 11$. 
The details of the encoding are given in \cref{app:encodings_matchings}.

\begin{theorem}
\cref{conjecture:hoffmanntoth_convex} is true for $n \le 11$.
\end{theorem}

\subsection{Empty $k$-Cycles}
\label{ssec:kyccles}

Another interesting strengthening of Rafla's conjecture is formulated in terms of empty $k$-cycles.
An \emph{empty $k$-cycle} is a plane cycle through $k$ vertices
of the drawing such that all remaining vertices lie on one common side.

\begin{conjecture}[{
\cite[Conjecture~5.1]{BFROS2024}}]
\label{conj:kcycles}
Every simple drawing of $K_n$ contains 
an empty $k$-cycle for every $k=3,\ldots,n$.
\end{conjecture}

The case $k=n$ 
coincides with Rafla's conjecture as empty $k$-cycles are precisely plane Hamiltonian cycles.
Empty 3-cycles coincide with empty triangles
and therefore the case $k=3$ is shown by Harborth~\cite{Harborth98}. 
Bergold et al.~\cite{BFROS2024}
showed that the conjecture is true 
for the subclass of convex drawings.
We used our SAT framework with parameter \verb|--emptycycles k|, where $k$ denotes the length of the empty cycle,
to verify the conjecture for up to $n = 10$. 
Details of the encoding are deferred to \cref{app:encodings_kcycles}.
\cref{tab:emptykcycles} shows the computing times for different parameters of $n$ and~$k$.
Note that the existence of empty $n$-cycles in all simple drawings of~$K_n$ imply the existence empty $n$-cycles in simple  drawings of~$K_{n+1}$. This is also reflected by computational data as the computing times for $k=n-1$ are comparably small. 

\begin{theorem}
\cref{conj:kcycles} holds for $n \leq 10$ for simple drawings. 
\end{theorem}

\begin{table}[htb]
\centering
\small
\def\arraystretch{1.2}
\begin{tabular}{c|rrrrrrrrr}
n/k	&3	&4	&5	&6	&7	&8	&9	&10	
\\
\hline
6	&$<$0.1	&$<$0.1	&$<$0.1	&$<$0.1	&	&	&	&	&	\\
7	&0.2	&1.0	&0.6	&0.1	&1.1	&	&	&	&	\\
8	&1.2	&10.3	&12.1	&15.0	&2.1	&21.4	&	&	&	\\
9	&3.7	&46.0	&337.0	&1248.7	&2732.7	&30.1	&1947.8	&	&	\\
10	&17.3	&467.5	&5804.5	&66600.3	
& 414212.9
&618399.9
&2968.2	&280384.4	&	\\
\end{tabular}

\caption{Computing time  
 for 
 empty $k$-cycles in simple drawings of~$K_n$ in CPU seconds.
}
\label{tab:emptykcycles}
\end{table}

\subsection{Uncrossed Edges}
\label{sec:uncrossededges}

Next we focus on edges which
are not crossed by any other edge.  In 1964,
Ringel~\cite{Ringel1964} proved that in every simple drawing of $K_n$ there
are at most $2n-2$ edges without a crossing. Later, Harborth and Mengersen
\cite{HarborthMengersen1974} studied the minimal number of uncrossed edges. 
They
showed that every simple drawing of $K_n$ with $n\le 7$ contains an uncrossed
edge and constructed simple drawings for $n \geq 8$ such that every edge is crossed.
The constructed drawings are strongly c-monotone but not convex. 
Moreover, they conjectured that crossing-maximal drawings, i.e., drawings with $\binom{n}{4}$ crossings, behave differently. 
\begin{conjecture}[
\cite{HarborthMengersen1992}]
    \label{conj:uncrossededgecrmax}
    Every crossing-maximal simple drawing of $K_n$ contains an \mbox{uncrossed} edge. 
\end{conjecture}

We used the \SAT{} framework with flags \verb|-aec| (all edges crossed) and \verb|-crmax| (crossing-maximal), we 
confirmed to verify \cref{conj:uncrossededgecrmax} for up to $n=16$. The computations for $n=14,15,16$ took about 3, 8, and 35 CPU hours, respectively.

\begin{theorem}
    \cref{conj:uncrossededgecrmax} is true for $n \leq 16$
\end{theorem}

When restricted to convex drawings, all
drawings with $n \leq 10$ have an uncrossed edge.
For $n = 11, \ldots, 21$, there are h-convex
drawings where every edge is crossed. 
The computations for $n=22$ timed out after 8 CPU days.
Kyn\v{c}l and Valtr~\cite{KynclValtr09} showed that for sufficiently large $n$ there is an h-convex drawing of $K_n$ without uncrossed edge.
We conjecture:

\begin{conjecture}
	\label{conjecture:hconvex_aec}
	For $n \ge 11$ there is an h-convex drawing of
	$K_n$ where every edge is crossed.
\end{conjecture}

\subsection{Quasiplanarity and Crossing Families}
\label{sec:crossingfamilies}
In simple drawings there are not only unavoidable plane structures, but also
unavoidable crossing structures. Given a drawing of $K_{n}$, a set of $k$
independent edges that cross pairwise is called a \emph{$k$-crossing family}.
A simple drawing without a $k$-crossing family is called
\emph{$k$-quasiplanar}. Brandenburg~\cite{Brandenburg2016} found a 3-quasiplanar drawing of $K_{10}$ which is moreover
h-convex and strongly c-monotone. 
As shown by Pitchanathan and Shannigrahi
\cite{PitchanathanShannigrahi2019}, no simple drawing of $K_{11}$ is
3-quasi\-planar, that is, every simple drawing of $K_{11}$ contains a
3-crossing family. 
We reproduced this results with the \SAT{} framework and parameter \verb|-crf 3|, which forbids crossing families of size~3.

Aichholzer et al.~\cite{AichholzerAurenhammerKrasser2001} used complete
enumeration to show that for $n \geq 9$ every geometric drawing of
$K_{n}$ has a $3$-crossing family and this is the smallest possible value. 
Moreover,  Aichholzer et al.\ \cite{AKSVV2021}
utilized \SAT{} solvers to show that 
for $n \geq 15$ every geometric drawing of
$K_{n}$
has a 4-crossing family.
Our framework found an h-convex drawing of $K_{16}$ without a 4-crossing
family after about 1 CPU day. We provide the example as supplemental data~\cite{supplemental_data_unblind}.

For crossing-maximal drawings there seem to be large crossing families:
\begin{conjecture}
    Every crossing-maximal drawing of $K_n$ contains a $\lfloor \frac{n}{2} \rfloor$-crossing family.
\end{conjecture}
The computations for $n=8,10$ took about 2 CPU seconds and 2 CPU hours, respectively.

\subsection{Generalized \texorpdfstring{Erd\H{o}s}{Erdös}--Szekeres Theorem}
\label{sec:unavoid}

The Erd\H{o}s--Szekeres theorem asserts that every sufficiently large point
set in the plane in general position contains a subset of $k$ points in convex
position (i.e., no point lies in the convex hull of the others). In the
context of geometric drawings this immediately implies that every sufficiently
large geometric drawing contains a subdrawing weakly isomorphic to the
\emph{perfect convex} drawing $\calC_k$. Here $\calC_k$ denotes the geometric graph
where the $k$ vertices are in convex position. A drawing of $K_k$ is weakly
isomorphic to $\calC_k$ if and only if there exists a labeling $1, \ldots, k$ of
the vertices such that $ac$ crosses $bd$ for all $1 \leq a<b<c<d \leq k$.
Recall that this defines the weak equivalence class of $\calC_k$ 
(\cref{prop:PRS_different_ATgraphs}).

Harborth and Mengersen~\cite{HarborthMengersen1992} showed that the Erd\H{o}s--Szekeres
theorem does not generalize to simple drawings. For every $k \geq 5$ the
\emph{perfect twisted} drawing $\calT_k$ does not contain a subdrawing weakly isomorphic
to~$\calC_5$. The perfect twisted drawing~$\calT_k$ is characterized by the existence of a
vertex labeling such that $ad$ crosses $bc$ if and only if
$1 \leq a<b<c<d \leq k$. 
\cref{fig:convex_C5_and_twisted_T5} shows drawings of $\calC_5$ and~$\calT_5$.

\begin{figure}[htb]
    \centering
    \includegraphics{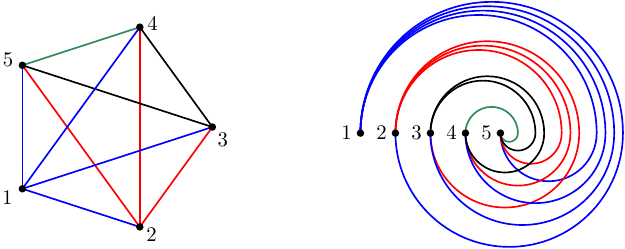}
    \caption{The two unavoidable configurations $\calC_5$ (left) and $\calT_5$ (right).}
    \label{fig:convex_C5_and_twisted_T5}
\end{figure}

Pach, Solymosi, and T\'oth \cite{PachSolymosiToth2003} proved an
Erd\H{o}s--Szekeres-type theorem for simple drawings: for all integers $a$ and $b$
there exists $N=N_{\simple}(\calC_a,\calT_b)$ such that for all $n\geq N$ every
simple drawing of $K_n$ contains $\calC_a$ or  
$\calT_b$ as a subdrawing. Suk and
Zeng~\cite{SukZeng2022} showed the currently best bound:
$ N_{\simple}(\calC_a,\calT_b) \le 2^{9(ab)^2 \log_2(a)\log_2(b)}$.
For the class of c-monotone and strongly c-monotone we analogously define $N_{\cmonotone}(\calC_a,\calT_b)$ and $N_{\strongcmonotone}(\calC_a,\calT_b)$, respectively.

When restricting the question to convex drawings, the perfect twisted drawing
$\calT_5$ (cf. $\obstructionconvexFiveA$)
does not occur. Hence every sufficiently large
convex drawing contains a perfect convex~$\calC_a$. 
For the convex and
\mbox{h-convex} drawings we can thus define the generalized
Erd\H{o}s--Szekeres numbers $N_{\convex}(\calC_a)$ and
$N_{\hconvex}(\calC_a)$, respectively, as the smallest integer~$n$ such that
every convex (resp.\ h-convex) drawing of~$K_n$ contains~$\calC_a$ as  a subdrawing. 
The bound
by Suk and Zeng~\cite{SukZeng2022} immediately gives
$N_{\hconvex}(\calC_a) \le N_{\convex}(\calC_a) \le~2^{O(a^2 \log a )}$.
Similarly, when restricting to generalized twisted drawings, the perfect convex drawing $\calC_5$ does not appear. Let $N_{\gentwisted}(\calT_a)$ be the smallest integer~$n$ such that
every generalized twisted drawing of~$K_n$ contains~$\calT_a$ as a subdrawing. It holds 
$N_{\gentwisted}(\calT_a) \le~2^{O(a^2 \log a )}$~\cite{SukZeng2022}.

Using our \SAT{} framework 
we determined the following values:
\begin{theorem}
    \label{theorem:C5T5}
    It holds  
    $N_{\simple}(\calC_5,\calT_5) = N_{\cmonotone}(\calC_5,\calT_5) = N_{\strongcmonotone}(\calC_5,\calT_5) = 13$,
    $N_{\convex}(\calC_5)= N_{\hconvex}(\calC_5)=11$,
    $N_{\gentwisted}(\calT_6) = 7 $, and
    $N_{\gentwisted}(\calT_7) = 10$.
\end{theorem}
The parameters \verb|-C a| and \verb|-T b| allow to forbid $\calC_a$ and~$\calT_b$, respectively.
Furthermore we introduced the parameter \verb|-X k| to forbid crossing-maximal subdrawings of~$K_k$. 
Since $C_5$ and $T_5$ are the only two crossing-maximal drawings of $K_5$, \verb|-X 5| and \verb|-C 5 -T 5| give equisatisfiable instances.
It took about 8 CPU seconds to prove $N_{\convex}(\calC_5) \le 11$
and 36 CPU hours to prove $N_{\simple}(\calC_5,\calT_5) \le 13$.
The tightness of the bounds is witnessed by examples showing $N_{\hconvex}(\calC_5) > 10$ and $N_{\strongcmonotone}(\calC_5,\calT_5) > 12$.
Moreover, the program found examples  witnessing
$N_{\hconvex}(\calC_6) > 21$, 
$N_{\simple}(\calC_6,\calT_5) > 23$,  
$N_{\simple}(\calC_5,\calT_6) > 16$,
and 
$N_{\gentwisted}(\calT_8) > 12$;
see the supplemental data~\cite{supplemental_data_unblind}.

The corresponding number for geometric drawings, $N_{\text{geom}}(\calC_k)$, is
the classical Erd\H{o}s--Szekeres number.
It is of order $2^{k+o(k)}$
\cite{Suk2017}.
This motivates the question, 
whether one of the above Ramsey numbers, in particular
$N_{\strongcmonotone}(\calC_a,\calT_a)$,
$N_{\gentwisted}(\calT_a)$,
and $N_{\hconvex}(\calC_a)$, can also be bounded by an exponential function
$c^a$ for some constant~$c$.

\subsection{Number of Empty Triangles}
\label{sec:emptytriang}

An \emph{empty triangle} in a simple drawing of $K_n$ is a triangle
induced by three vertices such that the interior of one of its two sides does
not contain any vertex.
Let $\triangle(n)$ denote
the minimum number of empty triangles in all simple drawings of $K_n$. Harborth \cite{Harborth98} proved
the bounds $2 \le \triangle(n) \le 2n-4$ and conjectured the upper bound to be optimal.

\begin{conjecture}[\cite{Harborth98}]
\label{conjecture:empty_triangles}
It holds $\triangle(n)=2n-4$.
\end{conjecture}

In a recent paper, García et al.~\cite{GarciaTejelVogtenhuberWeinberger2022} showed that generalized twisted drawings have exactly $2n-4$ empty triangles. 
Ruiz--Vargas \cite{RuizVargas2013} proved $\triangle(n) \ge \frac{2n}{3}$. 
The bound was further improved by Aichholzer et al.\
\cite{AichholzerHPRSV2015} to $\triangle(n) \ge n$.
Moreover, \'Abrego et al.\ \cite{AichholzerHPRSV2015,AbregoAFHOORSV2015} used
the database of rotation systems to show that $\triangle(n) = 2n-4$ holds for
$n \leq 9$.
We reproduced this result and further verified the conjecture for $n=10$ with our \SAT{} framework. The parameter \verb|-etupp k| asserts that there are at most $k$ empty triangles; see \cref{app:encodings_emptytriangles} for more details on the encoding. 
The computations for $n=7,8,9,10$ took about 10 CPU seconds, 25 CPU minutes,  46 CPU hours, and 16 CPU days respectively. 

\begin{theorem}
\cref{conjecture:empty_triangles} is true for $n \le 10$.
\end{theorem}

	Moreover, it is known that every convex drawing of $K_n$ contains at least
	$\triangle_{\convex}(n) \ge \Omega(n^2)$ empty triangles \cite{ArroyoMRS2017_pseudolines}, which is asymptotically tight as
	there exist geometric drawings with $\Theta(n^2)$ 
	empty triangles~\cite{BaranyFueredi1987}. Determining the minimum number of empty triangles in geometric drawings $\triangle_{\textnormal{geom}}(n)$ remains a 
	challenging problem cf.~\cite[Chapter~8.4]{BrassMoserPach2005} and \cite{ABHKPSVV2020_JCTA}.

\subsection{Enumeration}
\label{ssec:enumeration}

Last but not least, our framework can be used with the parameter \verb|-a|  to enumerate the solutions of a CNF. This in particular allows to count isomorphism classes of various classes of drawings; see \cref{table:numbers}. 
It is also possible to enumerate examples from the intersection of various classes by simply activating the desired flags, e.g., \verb|-hc -scm| restricts to rotation systems that are both h-convex and strongly c-monotone.
However, for enumeration and counting, classic approaches such as the Avis--Fukuda reverse search technique \cite{AvisFukuda1996} (as used in \cite{AbregoAFHOORSV2015}) seem to be better suited. For
SAT solvers that are based on the CDCL algorithm, one needs to add a clause 
for every found solution to prevent it from reappearing and therefore the size of the instance and the memory consumption grows linear in the number of objects.

\begin{table}[htb]
\centering
\small
\def\arraystretch{1.2}
\begin{tabular}{ r|rrrrrr }
	& all 
	& convex 
	& h-convex
	& c-monotone
	& str.\ c-monotone
	& gen.\ twisted
	\\
	\hline
	n = 4
	&2
	&2
	&2
	&2
	&2
	&1
	\\
	n = 5
	&5
	&3
	&3
	&5
	&5
	&1
	\\
	n = 6
	&102
	&16
	&15
	&102
	&95
	&3
	\\
	n = 7
	&11 556
	&139     
	&126    
	&11 556   
	&8 373
	&9 
        \\
	n = 8        
	&5 370 725$^\ast$
	&3 853   
	&3 394  
	&?
	&?  
	&32  
        \\    
	n = 9        
	&7 198 391 729$^\ast$
	& 215 105 
	& 183 982
	& ?
	& ? 
	& 115 
        \\
\end{tabular}
\caption{
Number of isomorphism classes of simple drawings of~$K_n$.
The entries marked with $\ast$
were computed in \cite{AbregoAFHOORSV2015} 
and were not verified by our framework.
}
\label{table:numbers}
\end{table}

\section{Discussion}
\label{sec:discussion}

To tackle a comprehensive list of conjectures on simple drawings of~$K_n$,
we developed SAT encodings for various properties 
and subclasses.
Our highly parameterized 
python framework 
allows to create a CNF instance such that the solutions correspond to drawings  with specified properties.
Modern SAT solvers 
are used to decide whether there exist such drawings.

To optimize the overall performance,
we experimented with the optimization tool SBVA \cite{SBVA2023}.
It introduces auxiliary variables into a CNF to reduce the number of clauses while preserving equisatisfiablity. 
Even though this approach does not always lead to a better performance of the SAT solver, 
we analyzed all instances where SBVA lead to a speedup and incorporated the learned patterns into our encoding.

Note that every CNF can be partitioned into independent parts (so-called cubes)
which can then be solved in parallel \cite{HKWB2012_cubes}.
However, it is a highly non-trivial task to find a good partition such that all cubes are of similar complexity
and a bad partitioning  
can lead to significantly worse computing times in total.

	\paragraph{Aesthetic drawings}

	Given the abstract graph of a planarization corresponding to drawing $D$ of a
	rotation system, we can find a plane drawing with explicit coordinates from
	the $m \times m$ integer grid by applying Schnyder's algorithm
	\cite{Schnyder1990}, where $m$ denotes the number of vertices of the planarization. Such drawings on the grid however
	do not look very aesthetic because vertices can be very close to edges and
	faces may be represented by arbitrary polygons.
	Instead we compute the embedding using iterated Tutte embeddings  to create an aesthetic visualization of
	the planarization, where every face is drawn convex and edges and faces are
	within a certain ratio. 
    We used a slight modification of the implementation from Felsner and
	Scheucher \cite{FelsnerScheucher2020}, see also the \verb|visualize.sage|  script in \cite{github_PCA}.
    Based on this drawing of the planarization,
	we visualize the original drawing of $K_n$ by drawing its edges
	as a Beziér curve,
    as illustrated in 
	the computer-generated graphic \cref{fig:computer_vis_5} in the appendix.
    Note that each original edge consists of a sequence of edges in the planarization.


\bibliography{references}

\appendix

\clearpage
\section{\SAT{} Encoding for Drawability of Pre-Rotation Systems 
}
\label{app:characterizationRS}

To give an independent proof of \cref{proposition:rotsys_classification_n6}, we model a CNF for deciding whether a given pre-rotation system is drawable.

In the following, we set up the desired Boolean formula that encodes a potential planarization -- if one exists.
The \emph{planarization} of a given drawing $\calD$ is the planar graph $G$ obtained by placing an auxiliary \emph{cross-vertex} at the position of every crossing point of $\calD$ and accordingly subdividing the two crossing edges. Each cross-vertex has degree four and increases the number of edges by two.
\cref{fig:planarization} gives an illustration.

\begin{figure}[htb]
    \centering
    \includegraphics[scale=0.9]{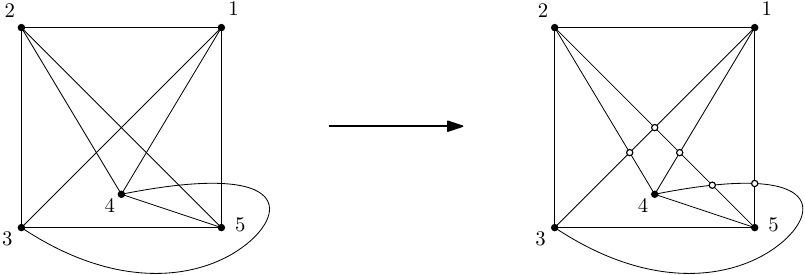}
    \caption{(left) A simple drawing of $K_5$ and 
    (right) its planarization. The auxiliary cross-vertices are depicted as circles.} 
    \label{fig:planarization}
\end{figure}

To decide whether a given pre-rotation system $\Pi$ is drawable we need to find a planarization of a drawing belonging to~$\Pi$.
By \cref{observation:basics}\eqref{item:obstructionfour_notdrawable} $\Pi$ is not drawable if it contains~$\obstructionFour$.
Hence we assume that $\Pi$ is $\obstructionFour$-free.
Now \cref{observation:basics}\eqref{item:obstructionfour_crossingsdetermined} implies that
the pairs of crossing edges $X$ can be inferred from the induced 4-tuples.
In particular, for each edge $e \in E$ 
we know which crossings~$X_e$ occur along it. 
The missing information to get a planarization 
-- if there exists one -- is the order of the crossings along the edges.

For every edge~$e = \{u,v\} \in E$ with $u<v$,
let $\Sigma_e$ be a permutation of the crossings $X_e$ along~$e$, which we extend by adding~$u$ as the first element and~$v$ as the last element.
We define the graph $G_{\Sigma}$ where the vertex set consists of the original elements $[n]$ and the crossings~$X$, 
and for every $e=\{u,v\}\in E$ with $u<v$,
we connect every pair of consecutive crossings from~$\Sigma_e$.
	
Our aim is to find an assignment for all permutations $\Sigma_e$, $e \in E$,
such that the corresponding graph $G_{\Sigma}$ is a planar graph. 
To ensure that the desired graph is planar, 
we use Schnyder's characterization of planar graphs~\cite{Schnyder1989}, which we describe in more detail in \cref{app:planar_enc}.

\begin{proposition}
\label{proposition:planarization_encoding_correct}
    Let $\Pi$ be a pre-rotation system.
    If $\Pi$ is drawable, 
    then there exists an assignment $\Sigma$ such that $G_\Sigma$ is planar.
    Moreover, 
    every planar $G_\Sigma$ 
    is the planarization of a drawing with rotation system~$\Pi$.
\end{proposition}

\begin{proof}
    The first part is straight-forward.
    If $\Pi$ is drawable, we consider a drawing and obtain $G_\Sigma$ as its planarization.
    
    For the second part,
    suppose that $G_{\Sigma}$ is a planar graph
    and consider a plane drawing $\calD$ of~$G_{\Sigma}$.
    We show that 
    \begin{enumerate}[(i)]
        \item 
        \label{item:propercrossings}the auxiliary cross-vertices are drawn as proper crossings; and 
        \item 
        \label{item:samers}the rotation system $\Pi_\calD$ of the drawing $\calD$ equals $\Pi$.
    \end{enumerate}
    
    To show \eqref{item:propercrossings},
    suppose that the cross-vertex $c$ of the edges $\{u,v\}$ and $\{u',v'\}$ of $K_n$ 
    is drawn as a touching, i.e., the endpoints of the edges are consecutive in the cyclic order around $c$. Without loss of generality we assume that they appear as follows $ u,v,v',u'$.
    We consider the planarization of the subdrawing of $K_n$ induced by the four vertices $ u',v',v,u$.
    Since a $K_4$ has at most one crossing, $c$ is the only crossing vertex in this subdrawing. 
    As illustrated in \cref{fig:propercrossing}
    not all edges of this planarization of $K_4$ can be drawn without crossings, 
    which is a contradiction.
    Hence, in~$\calD$ all crossing-vertices are drawn as proper crossings.
    
    Since the pairs of crossing edges are the same for $\Pi$ and $\Pi_\calD$, \eqref{item:samers}
    follows from the following \cref{prop:PRS_different_ATgraphs}.
\end{proof}

\begin{figure}[htb]
    \centering
    \includegraphics{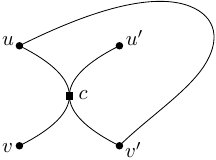}
    \caption{The vertices $u'$ and $v$ cannot be connected in a plane way.}
    \label{fig:propercrossing}
\end{figure}

\begin{proposition}
    \label{prop:PRS_different_ATgraphs}
    Let $\Pi$ and $\Pi'$ be two $\obstructionFour$-free pre-rotation systems on $[n]$.
    The pairs of crossing edges of $\Pi$ and $\Pi'$ coincide if and only if $\Pi$ is equal to $\Pi'$ or its reflection.
\end{proposition}

A weaker version of \cref{prop:PRS_different_ATgraphs} (restricted to rotation systems) was already proven by Kyn\v{c}l \cite[Proposition~2.1(1)]{Kyncl2020}.
	
\begin{proof}
    If $\Pi$ is equal to $\Pi'$ or its reflection,
    then the pairs of crossing edges clearly coincide (\cref{observation:basics}\eqref{item:obstructionfour_crossingsdetermined}).
    
    For the converse statement, 
    assume that $\Pi$ differs from $\Pi'$ and its reflection.
    We show that we find a set of 
    five vertices $I \subset V$ such that,
    when restricted to $I$,
    the (sub-)pre-rotation system
    $\Pi|_I$ still differs from $\Pi'|_I$ and its reflection.
    
    Case~1: there exists a vertex $v$ 
    such that the cyclic permutation $\pi_v$ differs from $\pi'_v$ and its reflection.
    Then there are four elements $J = \{a,b,c,d\}$ such that,
    when restricted to $J$, the cyclic permutation $\pi_v|_J$ still differs from $\pi'_v|_J$ and its reflection.
    Hence we can choose $I = \{v,a,b,c,d\}$.
    
    Case~2: for every vertex $v$ 
    the cyclic permutation
    $\pi_v$ is equal to $\pi'_v$ or its reflection.
    Since $\Pi$ differs from $\Pi'$ and its reflection,
    there are two vertices $v,w$
    such that $\pi_v$ equals $\pi'_v$
    and $\pi_w$ equals the reflection of~$\pi'_w$.
    By choosing any $a,b,c$ distinct from $v$ and $w$,
    we have five vertices $I=\{v,w,a,b,c\}$ 
    such that $a,b,c$ occur in the same order around~$v$
    and $a,b,c$ occur in the opposite order around~$w$.
    
    Using the program
    we enumerated all $\obstructionFour$-free pre-rotation systems on 5 vertices.
    This can be done with the command
    \begin{lstlisting}
        python rotsys.py 5 -v5 -a --nat --checkATgraphs
    \end{lstlisting}
    As checked by the flag \verb|--checkATgraphs|, 
    any such two pre-rotation systems 
    which are not reflections of each other
    have distinct pairs of crossing edges. 
    Hence the pairs of crossing edges for $\Pi$ and $\Pi'$ are different.
\end{proof}

\propclassificationRS*
	
\begin{proof}
    We combine the \SAT{} frameworks from \cref{sec:all_encoding} for pre-rotation systems 
    and from \cref{app:planar_enc} for planar graphs. 
    As a first step, we use
    the \SAT{} framework to enumerate the
    3 non-isomorphic pre-rotation systems on 4 elements using the command
    \begin{lstlisting}
	python rotsys.py 4 -v4 -a -l 
    \end{lstlisting}
    By \cref{observation:basics}\eqref{item:obstructionfour_notdrawable}
    we know that there are exactly two non-isomorphic rotation systems on 4 vertices and the non-drawable pre-rotation system $\obstructionFour$.
    Next, we enumerate the
    7 non-isomorphic pre-rotation systems on 5 vertices,
    that do not contain~$\obstructionFour$.
    \begin{lstlisting}
        python rotsys.py 5 -v5 -a -l -r2f all5.json0 
    \end{lstlisting}
    Our drawability framework shows 
    that exactly five of them are drawable. 
    \begin{lstlisting}
	sage rotsys_draw.sage all5.json0  
    \end{lstlisting}
    See \cref{fig:computer_vis_5}
    for the computer generated\footnote{%
    We used iterated Tutte embeddings \cite{FelsnerScheucher2020} to automatically generate aesthetic visualizations of the planarizations.   
 Edges
	are drawn as a Beziér curve to make the visualization even more appealing.} 
    drawings of $K_5$.
    The
    two configurations~$\obstructionFiveA$ and~$\obstructionFiveB$ are not drawable.
    Last but not least, 
    we use the framework with the command 
    \begin{lstlisting}
        python rotsys.py 6 -a -l -r2f all6.json0
    \end{lstlisting}
    to enumerate 
    102 non-isomorphic $\obstructionFour$, $\obstructionFiveA$, $\obstructionFiveB$- free pre-rotation systems on 6 elements.
    With the drawing framework we verify 
    that all of them are drawable.
    \begin{lstlisting}
	sage rotsys_draw.sage all6.json0
    \end{lstlisting}
    This completes the proof of \cref{proposition:rotsys_classification_n6}.
\end{proof}
	
\begin{figure}[htb]
    \centering
    \includegraphics[width=0.95\columnwidth]{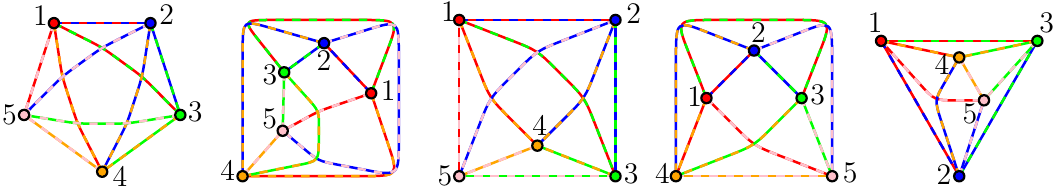}
    \caption{Computer generated drawings of the five types of rotation systems on 5 elements.}
    \label{fig:computer_vis_5}
\end{figure}

\section{\SAT{} encoding for planar graphs}
\label{app:planar_enc}

Schnyder \cite{Schnyder1989}
characterized planar graphs as graphs whose incidence poset is of order dimension at most three. 
That is,
a graph $G=(V,E)$ is planar if and only if there exist three total orders $\prec_1,\prec_2,\prec_3$ such that for every edge $\{u,v\} \in E$ and every vertex $w \in V \setminus \{u,v\}$ it holds $u \prec_i w$ and $v \prec_i w$ for some $i \in \{1,2,3\}$.
This characterization allows to encode planarity in terms of a Boolean formula.
Besides the fact that this approach can be used for planarity testing, 
it can also be used to enumerate all planar graphs with a specified property or 
to decide that no planar graphs with this property exists.
It is worth noting that 
certain subclasses of planar graphs 
such as outer-planar graphs
can be characterized in a similar fashion; see \cite{FelsnerTrotter05}.
For further literature and other planarity encodings suitable for \SAT{}-based investigations we refer the interested reader to Kirchweger et al.\
\cite{KirchwegerSS2023}.

\medskip
	
For every pair of vertices $u,v \in V$ with $u \neq v$
we use Boolean variables $A_{u,v}$ 
to encode whether $\{u,v\}$ is an edge in the graph
and, for any $i \in \{1,2,3\}$, we use Boolean variables $B_{i,u,v}$ to encode whether $u \prec_i v$.
To assert that the $Y$ variables indeed model a total order,
we must ensure the linear order and anti-symmetry with the constraints 
$B_{i,u,v} \vee B_{i,v,u}$ and $\neg B_{i,u,v} \vee \neg B_{i,v,u}$
and transitivity with the constraints 
$\neg B_{i,u,v} \vee \neg B_{i,v,w} \vee B_{i,u,w}$
for every $i \in \{1,2,3\}$ and distinct $u,v,w \in V$.
	
To ensure Schnyder's planarity conditions,
we introduce auxiliary variables $B_{i,u,v,w}$ to encode whether $u \prec_i w$ and $v \prec_i w$ holds for any three distinct variables $u,v,w$ and $i \in \{1,2,3\}$
(that is, $B_{i,u,v,w} = B_{i,u,w} \vee B_{i,v,w}$)
and then assert 
$\neg A_{u,v} \vee \bigvee_{i=1}^3 B_{i,u,v,w}$.

\section{Detailled Encodings}
\label{app:encodings}

\subsection{Plane Substructures}
\label{app:encodings_planesubstructures}
To assert that 
a subset of the edges $E' \subset E(K_n)$ 
does forms a plane subdrawing,
we set the corresponding crossing variables to $\false$, that is, we have a unit-clause 
$ \neg C_{e,f}$
for every pair of non-adjacent edges $e,f \in E'$.
Similarly, we can assert that $E'$ does not form a plane subdrawing with the clause
 $\bigvee_{e,f \in E' \colon e \cap f = \emptyset}
    C_{e,f}$.

To assert that there is no plane Hamiltonian cycle,
we assert that 
for every permutation $\pi$ 
there is at least one crossing pair in the set of edges $E_\pi = \{(\pi(i),\pi(i+1)) : i \in [n]\}$,
where $\pi(n+1)=\pi(1)$.
The encoding comes with $(n-1)!$ clauses
and is only suited for small values of~$n$.
Similarly, we can deal with plane Hamiltonian subdrawings on $2n-3$ edges: We assert that there is at least one crossing formed by every edge set $E' \in \binom{E(K_n)}{2n-3}$ which contains the edges $E_\pi$ of some Hamiltonian cycle~$\pi$.

\subsection{Extending Plane Matchings}
\label{app:encodings_matchings}

To be more specific about the encoding for testing \cref{conjecture:hoffmanntoth_convex}, that is, searching plane Hamiltonian cycles for a prescribed matching: 
we assume towards a contradiction that there exists a convex drawing of $K_n$ 
which has a plane matching $M = \{ \{1, 2\}, \ldots , \{2k-1, 2k\} \}$ with $2k\leq n$ such that every plane Hamiltonian cycle $C$ crosses at least one edge of~$M$.
Note that $M$ and $C$ may share edges.
As we fix the vertices of~$M$,  we cannot further assume without loss of generality that the rotation system
is natural for $k \ge 2$.  
However, to speed up the computations, 
we break further symmetries in the search space:
\begin{itemize}
    \item for every edge $\{u, u+1\} \in M$: $X_{1,2,u,u+1} = \true$, that is, 1 sees 2, $u$, $u+1$ in this order;
    \item for every two edges $\{u, u+1\}, \{u' , u'+1\} \in M$ with $u < u'$: $X_{1,2,u,u'} = \true$, that is, 1 sees 2, $u$, $u'$ in this order;
    
    \item for every $x, y \in [n]\backslash \{1, \ldots, 2k\}$ with $x < y$: $X_{1,2,x,y} = \true$, that is, 1 sees 2, $x$, $y$ in this order.
\end{itemize}

\subsection{Empty $k$-Cycles}
\label{app:encodings_kcycles}

Next we outline our encoding for empty $k$-cycles.
The vertices $\pi_1,\ldots,\pi_k$ form an empty $k$-cycles
if and only if it forms a plane $k$-cycle and all other vertices lie on one common side of the cycle.
Let $E_\pi = \{\{\pi_i,\pi_{i+1}\}: i =1,\ldots,k\}$ with $\pi_{k+1}=\pi_1$ be the edges of the cycle. The first part can easily be formulated as using our auxiliary crossing variables.
For the second part we need to introduce auxiliary variables
$W_{\pi,p,q}$ that indicate whether the edge $pq$ intersects an odd number of edges from~$E_\pi$. 
Since we only consider simple drawings, the $W$ variable indicates whether $p$ and $q$ lie in distinct sides.
To assign the $W_{\pi,p,q}$ variable, 
we consider each of the $2^k$ cases how $pq$ may intersect~$E_\pi$ and assign \true{} or \false, 
depending on whether the number of crossings with $E_\pi$ even or odd.
Note that this gives $n^k 2^k$ clauses in total.
Ultimately, we forbid that $\pi_1,\ldots,\pi_k$ forms an empty $k$-cycles using the clauses
\[
\bigvee_{e,f \in E_\pi} C_e,f  
\vee  
\bigvee_{p,q } W_{\pi,p,q}.
\]

\subsection{Empty Triangles}
\label{app:encodings_emptytriangles}

Empty triangles are exactly the empty 3-cycles and hence
the encoding from \cref{app:encodings_kcycles} can be used.
However, since it comes with $n^5$ auxiliary variables, 
we now describe a more compact encoding that suffices with only $n^3$ auxiliary variables.

We 
distinguish the two sides of a triangle. 
For a triangle spanned by three distinct vertices $a,b,c$, 
one of its sides sees $a,b,c$ in clockwise order and the other one in counterclockwise order. 
We denote by $S_{a,b,c}$ the side of the triangle which has $a,b,c$ in counterclockwise order.
In each of the drawings in \cref{fig:triangle_abc}, $S_{a,b,c}$ is the bounded side and $S_{a,c,b}$ is the unbounded side.
For every three distinct indices $a,b,c$ we introduce a Boolean variable $E_{a,b,c}$ which is $\true$ if the side $S_{a,b,c}$ is empty.
To assert the $E$ variables, we introduce auxiliary variables.
For a point $d \in [n] \setminus \{a,b,c\}$, we define $E_{a,b,c}^d= \false$ if $S_{a,b,c}$ contains $d$.
For four elements, there exist eight rotation systems, the four rotation systems in \cref{fig:rs_n4_noniso} and their reflections.
Hence, we set the variable $E_{a,b,c}^d$ to $\false$ if and only if 
one of the four cases depicted in \cref{fig:triangle_abc} occurs. 

\begin{figure}[htb]
    \centering
    \includegraphics[scale =0.9]{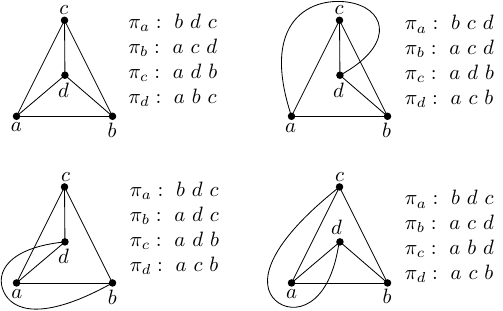}
    \caption{The four cases where $S_{a,b,c}$ contains~$d$.}
    \label{fig:triangle_abc}
\end{figure}

The side $S_{a,b,c}$ is empty if and only if
no point $d$ lies in $S_{a,b,c}$, hence 
\begin{align*}
    E_{a,b,c} = \bigwedge_{d \in [n] \setminus \{a,b,c\}} E_{a,b,c}^d.
\end{align*}

\end{document}